\documentclass[a4paper, 12pt]{article}

\usepackage{mathptmx}       % selects Times Roman as basic font
\usepackage{helvet}         % selects Helvetica as sans-serif font
\usepackage{courier}        % selects Courier as typewriter font
\usepackage{type1cm}        % activate if the above 3 fonts are
\usepackage{makeidx}         % allows index generation
\usepackage{graphicx}        % standard LaTeX graphics tool
\usepackage{multicol}        % used for the two-column index
\usepackage[bottom]{footmisc}% places footnotes at page bottom
\usepackage{url}

\usepackage{amsmath}
\usepackage{mathrsfs}
\usepackage{amsfonts}
\usepackage{authblk}

\newtheorem{theorem}{Theorem}
\newtheorem{lemma}[theorem]{Lemma}

\newtheorem{remark}[theorem]{Remark}

\DeclareMathOperator{\cone}{cone}
\DeclareMathOperator{\diff}{d}

\DeclareMathOperator{\conv}{conv}

\DeclareMathOperator{\1}{1}

\newcommand{\fF}{\frak F}
\newcommand{\fG}{\frak G}
\newcommand{\fS}{\frak S}
\newcommand{\cM}{\mathscr M}

\newcommand{\calF}{\mathcal F}
\newcommand{\calG}{\mathcal G}
\newcommand{\calS}{\mathcal S}

\newcommand{\comment}[1]{}

\newcommand{\N}{\mathbb{N}}
\newcommand{\R}{\mathbb{R}}

\makeindex             

\begin{document}

\title{A New Approach to Model Free Option Pricing}
\author[1]{Raphael Hauser\thanks{Andrew Wiles Building, Radcliffe Observatory Quarter, Woodstock Road, Oxford OX2 6GG,
United Kingdom, hauser@maths.ox.ac.uk}}
\author[2]{Sergey Shahverdyan\thanks{Andrew Wiles Building, Radcliffe Observatory Quarter, Woodstock Road, Oxford OX2 6GG,
United Kingdom, shahverdyan@maths.ox.ac.uk}}
\affil[1]{Mathematical Institute, University of Oxford}
\affil[2]{Mathematical Institute, University of Oxford}
\date{}
\maketitle

\abstract{}
In this paper we introduce a new approach to model-free path-dependent option pricing. We first introduce a general duality result for linear optimisation problems over signed measures introduced in \cite{hauser} and show how the the problem of model-free option pricing can be formulated in the new framework. We then introduce a model to solve the problem numerically when the only information provided is the market data of vanilla call or put option prices. Compared to the common approaches in the literature, e.g. \cite{hen-lab}, the model does not require the marginal distributions of the stock price for different maturities. Though the experiments are carried out for simple path-dependent options on a single stock, the model is easy to generalise for multi-asset framework.

%%%%%%%%%%%%%%%%%%%%%%
\section{Introduction}\label{introduction}
%%%%%%%%%%%%%%%%%%%%%%
In this paper we introduce a new approach on model-free pricing of discrete path-dependent options. There is a considerable amount of research done on the theory of the problem, but, to our knowledge, not much results are known for the numerical solution for the general case. This work is based on a general duality relation introduced in \cite{hauser}. This duality relation serves us twice in this work. First we show how the optimisation problem of finding model-free lower and upper bounds on path dependent options can be formulated in their framework, which allows one to get the dual problem straightaway from the primal.\\
Next, we apply the so-called 'Constraints via Bounds on Integrals' model to solve the optimisation problem numerically. This model allows one to construct constraints for marginal distributions, which can be obtained by calibrating the empirical risk-neutral distributions. But instead we use market data of European  call or put prices as constraints. This allows us to compute the tightest possible bounds in the model-free framework, i.e., given only the market information, any narrower bound has to consider constraints other than those provided by  the market data and that the stock price process is a martingale.

%%%%%%%%%%%%%%%%%%%%%%
\section{Problem Formulation and Duality}\label{duality}
%%%%%%%%%%%%%%%%%%%%%%
Consider an exotic path-dependent option that depends on the value of a stock price at some discrete times $t_1 < t_2 < \dots < t_n$. Let $X = (X_1, \dots, X_n)$ be a $n$-dimensional random variable denoting stock prices at times $(t_1, \dots, t_n)$ respectively, and let $h(X_1, \dots, X_n)$ be the payoff of the given option.\\
In practice, in a no-arbitrage framework one assumes that the stock price follows some specific model, one calibrates this model to market data and finds the fair price of an option, i.e., find a risk neutral probability measure $Q$ on $\R^n$ and calculate the price as
\begin{equation}\label{eq:optionPricingPrimalMain}
\mathbb E_Q[h] = \int_{\R^n}h(x_1, \dots, x_n)\diff Q(x_1, \dots, x_n).
\end{equation}
Under $Q$ the stock price process $(X_i)_{i = 1, \dots, n}$ is assumed to be a discrete martingale.\\
The problem of model-free option pricing is to find bounds on the prices of the option without assuming any particular model for the stock price process, i.e., find maximum and minimum values over all models that are consistent with the observed market data. Assuming knowledge of market call option prices, one can recover a distribution function $\calF_i$ of $X_i$ for every $i = 1, \dots, n$ (see, e.g., \cite{monteiro}), which will represent the marginal distributions of a risk neutral equivalent probability measure of $X$. In this framework, denoting by $\mathcal M_n(\calF_1, \dots, \calF_n)$ the set of martingale measures with marginals $\calF_1, \dots, \calF_n$, the problem becomes
\begin{equation}\label{MFOP:Primal}
\inf_{\calF\in\mathcal M_n(\calF_1, \dots, \calF_n)}\int_{\R^n}h(x_1, \dots, x_n)\diff \calF(x_1, \dots, x_n).
\end{equation}
We call the above problem the primal problem.\\
For any $i = 1, \dots, n$ define the set of $F_i$ integrable functions from $\R\rightarrow\R$ by $\mathcal H_i$ and for any $j = 1, \dots, n - 1$ the set of bounded measurable functions from $\R^j\rightarrow\R$ by $\mathcal T_j$. For any $h_i\in\mathcal H_i$, $i = 1, \dots, n$ and $g_j\in \mathcal T_j$, $j = 1, \dots, n - 1$ consider the function $\Psi_{(h_i),(g_j)}(x_1, \dots, x_n) : \R^n\rightarrow \R$ defined by
\begin{equation}
\Psi_{(h_i),(g_j)}(x_1, \dots, x_n) = \sum_{i = 1}^n h_i(x_i) + \sum_{j = 1}^{n-1}g_j(x_1, \dots, x_j)(x_{j+1} - x_j).
\end{equation}
Consider the following optimisation problem:
\begin{align}\label{MFOP:Dual}
\sup_{(h_i\in\mathcal H_i)_1^n,(g_j\in \mathcal T_j)_1^{n-1}}&\sum_{i = 1}^n\int_\R h_i(x_i)\diff \calF_i(x_i)\\
\text{s.t. }&\Psi_{(h_i),(g_j)}(x_1, \dots, x_n) \leq h(x_1, \dots, x_n).
\end{align}
We have the following theorem.\\
\begin{theorem}[Theorem 1, \cite{beig}]\label{beigTheorem1}
Assume $\calF_1, \dots, \calF_n$ are Borel probability measures on $\R$ so that $\mathcal M(\calF_1, \dots, \calF_n)$ is non-empty. Let $h:\R^n\rightarrow (-\infty, \infty]$ be a lower semi-continuous function so that
\begin{equation*}
h(x_1, \dots, x_n) \geq -K\cdot(1 + |x_1| + \dots + |x_n|),
\end{equation*}
on $\R^n$ for some constant $K$. Then the duality gap between the primal problem \eqref{MFOP:Primal} and the dual \eqref{MFOP:Dual} equals zero. Moreover, the primal optimal value is obtained, i.e. there exists a martingale measure $\calF\in\mathcal M(\calF_1, \dots, \calF_n)$ such that the optimal value of \eqref{MFOP:Primal} is equal to $\mathbb E_{\calF}[h]$.
\end{theorem}
A proof of the theorem can be found in $\cite{beig}$.\\
In order to solve the dual problem numerically, we need to make some discretisation of this infinite dimensional problem. One way of doing so (as is described in \cite{beig}) is to decompose the functions $(h_i)_1^n$ and $(g_j)_1^{n-1}$ over a finite dimensional basis. Since $(g_j)_1^{n-1}$ are assumed to be continuous, we can choose a polynomial basis, i.e.
\begin{equation*}
g_j(x_1, \dots, x_j) = \sum_{k}g_j^k e_j^k(x_1, \dots, x_j),
\end{equation*}
where $e_j^k$ are functions of the basis, and $g_j^k$ are the coordinates of $g_j$ with respect to the chosen basis.\\
For functions $(h_i)_1^n$ we can choose call option payoffs as basis
\begin{equation*}
h_i(x_i) = \sum_k h_i^k(x_i - k^b)^+.
\end{equation*}
With this approximation the dual problem reduces to a linear programming problem that can be solved via interior-point methods, the simplex algorithm or another algorithm of choice.\\
In this work we take a different approach to the numerical solution of the problem of arbitrage-free pricing. The approach is based on the model described in \cite{hauser}, which shall be presented in the next section.\\

%%%%%%%%%%%%%%%%%%%%%%
\section{A General Duality Relation}\label{duality}
%%%%%%%%%%%%%%%%%%%%%%
In this section we will see that the duality relation between \eqref{MFOP:Primal} and \eqref{MFOP:Dual} introduced by \cite{beig} in a special case of a general duality relation for mass transportation problems presented in \cite{hauser}, thus yielding an alternative derivation of \eqref{MFOP:Dual}.\\
First we will describe the duality results, and then will show how to obtain the duality results presented in the previous section using this framework. In this section we also present a model that lends itself to the numerical solution of the resulting infinite-dimensional optimisation problems, which have special structure.\\
Let $(\Phi,\fF)$, $(\Gamma,\fG)$ and $(\Sigma,\fS)$ 
be complete measure spaces, and let 
$A:\,\Gamma\times\Phi\rightarrow\R$, 
$a:\,\Gamma\rightarrow\R$, 
$B:\,\Sigma\times\Phi\rightarrow\R$, 
$b:\,\Sigma\rightarrow\R$, and 
$h:\,\Phi\rightarrow\R$ 
be bounded measurable functions on these spaces and the corresponding product 
spaces. Let $\cM_{\fF}$, $\cM_{\fG}$ and $\cM_{\fS}$ be the set of signed 
measures with finite variation on $(\Phi,\fF)$, $(\Gamma,\fG)$ and 
$(\Sigma,\fS)$ respectively.\\
Consider the following pair of optimisation problems over $\cM_{\fF}$ and 
$\cM_{\fG}\times\cM_{\fS}$ respectively, 
\begin{align}\label{eq:generalDualityPrimal}
\sup_{\calF\in\cM_{\fF}}\,&\int_{\Phi}h(x)\diff \calF(x)\\
\text{s.t. }&\int_{\Phi}A(y,x)\diff \calF(x)\leq a(y),\quad(y\in\Gamma),\nonumber\\
&\int_{\Phi}B(z,x)\diff \calF(x)= b(z),\quad(z\in\Sigma),\nonumber\\
&\calF\geq 0,\nonumber
\end{align}
and
\begin{align}\label{eq:generalDualityDual}
\inf_{(\calG,\calS)\in\cM_{\fG}\times\cM_{\fS}}
\,&\int_{\Gamma}a(y)\diff \calG(y)+\int_{\Sigma}b(z)\diff\calS(z),\\
\text{s.t. }&\int_{\Gamma}A(y,x)\diff\calG(y)+\int_{\Sigma}B(z,x)\diff\calS(z)\geq h(x),\quad(x\in\Phi),\nonumber\\
&\calG\geq 0.\nonumber
\end{align}
The infinite-programming problems \eqref{eq:generalDualityPrimal} and \eqref{eq:generalDualityDual} are duals of 
each other. 

\begin{theorem}[Weak Duality, Theorem 1, \cite{hauser}]\label{weak duality}
For every \eqref{eq:generalDualityPrimal}-feasible measure $\calF$ and every 
\eqref{eq:generalDualityDual}-feasible pair $(\calG,\calS)$ we have 
\begin{equation*}
\int_{\Phi}h(x)\diff\calF(x)\leq \int_{\Gamma}a(y)\diff\calG(y)+
\int_{\Sigma}b(z)\diff\calS(z).
\end{equation*}
\end{theorem}

\begin{proof}
Using Fubini's Theorem, we have 
\begin{align*}
\int_{\Phi}h(x)\diff\calF(x)&\leq
\int_{\Gamma\times\Phi}A(y,x)\diff(\calG\times\calF)(y,x)+
\int_{\Sigma\times\Phi}B(z,x)\diff(\calS\times\calF)(z,x)\\
&\leq\int_{\Gamma}a(y)\diff\calG(y)+
\int_{\Sigma}b(z)\diff\calS(z).
\end{align*}
\end{proof}

In various special cases, such as those discussed in  \cite{hauser}, strong duality is known to hold subject to regularity assumptions, that is, the optimal values of \eqref{eq:generalDualityPrimal} and \eqref{eq:generalDualityDual} coincide. Another special case under  which strong duality applies is when the measures $\calF$, $\calG$ and $\calS$ have densities in appropriate Hilbert spaces, see the forthcoming DPhil thesis of the second author \cite{sergey}. 
Note also that the quantifiers in the constraints can be weakened if the set of allowable measures is restricted. For example, if $\calG$ is restricted to lie in a set of measures that are absolutely continuous with respect to a fixed measure $\calG_0\in\cM_{\fG}$, then the quantifier $(y\in\Gamma)$ can be weakened to $(\calG_0\text{-almost all }y\in\Gamma)$.

\subsection{Application of the General Duality Result to Model Free Option Pricing}
In this subsection we show that the duality result for no-arbitrage bounds of path-dependent options can also be obtained in this framework.\\
In the original Problem \eqref{MFOP:Primal} we have two types of constraints. The first type is the requirement of the joint probability measure $\calF$ to have marginal distributions $\calF_1, \dots, \calF_n$. The second type of the constraint is for $\calF$ to be a martingale measure.\\
The marginal constraints can be formulated as follows
\begin{equation}
\text{s.t. }\int_{\R}\1_{\{x_i\leq z\}}(z,x)\diff \calF(x)
=F_i(z),\quad(z\in\R,i = 1, \dots, n)\\
\end{equation}
For the martingale measure constraints, observe, that $\calF$ is a martingale measure if and only if for every $k = 1, \dots, n - 1$ and every Borel measurable set $\mathcal B\subseteq\R^k$
\begin{equation}
\int_{\R^n}\1_{\mathcal B}(x_1, \dots, x_k)(x_{k + 1} - x_k)\diff \calF(x_1, \dots, x_n) = 0,
\end{equation}
which is equivalent to (\cite{beig}, Lemma 2.3):
for every $z, y\in\R^k$ such that $z\leq y$
\begin{equation}\label{eq:martingaleEquivalent}
\int_{\R^n}\1_{z\leq x|_k \leq y}(x_1, \dots, x_k)(x_{k + 1} - x_k)\diff \calF(x_1, \dots, x_n) = 0,
\end{equation}
where by $x|_k$ we denote the first $k$ coordinates of $x$ and the comparison of the vectors is component-wise.\\
Now we are ready to formulate the Problem \eqref{MFOP:Primal} in the new framework.\\
We aim to find the dual problem of the following optimisation problem.
\begin{align*}
\inf_{\calF}\,&\int_{\R^n}\phi(x)\diff\calF(x)\\
\text{s.t. }&\int_{\R^n}\1_{\{x_i\leq t\}}(t,x)\diff\calF(x)=F_i(t),\quad(t\in\R,i\in\N_n)\\
& \int_{\R^n}\1_{z|_k\leq x|_k\leq y|_k}(z, y, x)(x_{k+1} - x_k)\diff\calF(x) = 0,\quad(y,z\in\R^n,k<n)\\
&\calF\geq 0. 
\end{align*}
Denote $\Phi = \R^n$ and $\Sigma = \left\{1, \dots, 2n - 1\right\}\times\R\times\R^n\times\R^n$.\\
We construct the function $b(k, t, z, y) = F_k(t)$ if $k\leq n$, and $b(k, t, z, y) = 0$ for $k > n$.\\
Further, for $k\leq n$ we have $B(k, t, z, y, x) = \1_{x_k\leq t}$. For $k = 1, \dots, n - 1$ we have
\begin{equation*}
B(k + n, t, z, y, x) = \1_{z|_k \leq x|_k\leq y|_k}(x_{k+1} - x_k)
\end{equation*}
Now we construct the dual problem.\\
As a objective function we have
\begin{equation*}
\int_{\Sigma}b(k, t, z, y)\diff\calS(k, t, z, y) = \sum_{k = 1}^n\int_{\R\times\R^n\times\R^n}F_k(t)\diff\bar\calS_k(t, z, y).
\end{equation*}
For $k = 1, \dots, n$ define the measures $\calS_k$ on Borel sigma-algebra $\mathcal B(\R)$ such that $\calS_k(D) = \bar \calS_k(D\times \R^n\times \R^n)$ for every $D\in\mathcal B(\R)$ .
As for the marginal problem, the signed measures $\calS_k$ being of finite variation, the functions 
$S_k(t)=\calS((-\infty,t])$ and the limits $s_k=\lim_{t\rightarrow
\infty}S_k(t)=\calS((-\infty,+\infty))$ are well defined and finite. 
Furthermore, using $\lim_{t\rightarrow-\infty}F_k(t)=0$ and \\$\lim_{t\rightarrow+\infty}F_k(t)=1$, we have 
\begin{align*}
\sum_{k=1}^n\int_{\R}F_k(t)\diff\calS_k(t)&=
\sum_{k=1}^n\left(F_k(t)S_k(t)|^{+\infty}_{-\infty}-\int_{\R}S_k(t)\diff F_k(t)
\right)\\
&=\sum_{k=1}^n s_k - \sum_{k=1}^n\int_{\R}S_k(t)\diff F_k(t)\\
&=\sum_{k=1}^n\int_{\R}(s_k-S_k(t))\diff F_k(t),
\end{align*}
and likewise for constraints. Writing $h_k(t) = s_k - S_k(t)$, $(k = 1, \dots, n)$, the objective function becomes
\begin{equation*}
\sum_{k=1}^n\int_{\R}h_k(t)\diff F_k(t),
\end{equation*}
and the left hand side of the constraints become
\begin{align*}
\int_{\Sigma}B(k, t, z, y, x)&\diff\calS(k, t, z, y) \\
&= \sum_{k = 1}^n\int_{\R\times\R^n\times\R^n}\1_{x_k\leq t}\diff\bar\calS_k(t, z, y)\\
 &+ \sum_{k = 1}^{n-1}\int_{\R\times\R^n\times\R^n}\1_{z|_k \leq x|_k\leq y|_k}(x_{k+1} - x_k)\diff\bar\calS_{n + k}(t, z, y)\\
 &=\sum_{k=1}^n h_k(x_k) + \sum_{k = 1}^{n - 1}g_k(x_1, \dots, x_k)(x_{k+1} - x_k),
\end{align*}
where $g_k(x_1, \dots, x_k) = \int_{\R^n\times\R^n}\1_{z|_k \leq x|_k\leq y|_k}(x_{k+1} - x_k)\diff\calS_{k + n}(z, y)$, $k = 1, \dots, n - 1$. Measures $\calS_{n + k}$, $k = 1, \dots, n - 1$ are defined such that $\calS_{k + n}(D) = \bar \calS_{k + n}(\R\times D)$ for every $D\in\mathcal B(\R^n\times\R^n)$.\\ 
We obviously have that the functions $g_k$, $k = 1, \dots, n - 1$ are bounded. So we get the following dual formulation of our primal problem.
\begin{align}\label{eq:firstDualFormulationOfMDOpPricing}
\sup_{h_1,\dots,h_n}\,&\sum_{k=1}^n\int_{\R}h_k(\tau)\diff F_k(\tau)\\
\text{s.t. }&\sum_{k=0}^n h_k(x_k) + \sum_{k = 1}^{n - 1}g_k(x_1, \dots, x_k)(x_{k+1} - x_k)\leq h(x),\quad
(x\in\R^n).\nonumber
\end{align}
Since the decision variables $g_1, \dots, g_{n-1}$ have zero coefficient in the objective function, we get that the dual problem \eqref{eq:firstDualFormulationOfMDOpPricing} is equivalent to
\begin{equation}\label{dualBeiglbock}
\sup\left\{\sum_{k = 1}^n\mathbb E_{F_k}\left[h_k\right] : \exists g_1, \dots g_{n-1} \text{ s.t. } \Psi_{(h_k),(g_j)}\leq\phi\right\},
\end{equation}
where 
\begin{equation*}
\Psi_{(h_k),(g_j)}(x_1, \dots, x_n) = \sum_{k=0}^n h_k(x_k) + \sum_{k = 1}^{n - 1}g_k(x_1, \dots, x_k)(x_{k+1} - x_k)
\end{equation*}
and $\mathbb E_{F_k}(h_k(t)) = \int_{\R}h(t)\diff F_k(t)$.\\ 

Note that (\ref{dualBeiglbock}) is exactly the same dual as  \eqref{MFOP:Dual}.

\subsection{Constraints via Bounds on Integrals}\label{modelGeneral}
In this section we present a particular case of the general duality result discussed at the beginning of this section, specifically, where we have finitely many constraints. In the next section we show how to find no-arbitrage bounds of the options with piece-wise linear payoffs using this model.\\
Let us assume that $\Phi$ is decomposed into a partition $\Phi=\bigcup_{i=1}^k\Xi_i$ of polyhedra $\Xi_i$ with nonempty interior. Each polyhedron has a primal description in terms of generators, 
\begin{equation*}
\Xi_i=\conv(q^i_1,\dots,q^i_{n_i})+\cone(r^i_1,\dots,r^i_{o_i})
\end{equation*}
where $\conv(q^i_1,\dots,q^i_{n_i})$ is the polytope with vertices $q^i_n\in\R^n$, and 
\begin{equation*}
\cone(r^i_1,\dots,r^i_{o_i})=\left\{\sum_{m=1}^{o_i}\xi_m r^i_m:\,
\xi_m\geq 0\;(m\in\N_{o_i})\right\} 
\end{equation*}
is the polyhedral cone with recession directions $r^i_m\in\R^n$. Each polyhedron also has a dual description 
in terms of linear inequalities, 
\begin{equation*}
\Xi_i=\bigcap_{j=1}^{k_i}\left\{x\in\R^n:\, \langle f^i_{j}, x\rangle\geq \ell^i_j\right\}, 
\end{equation*}
for some vectors $f^i_j\in\R^n$ and bounds $\ell^i_j\in\R$. The main case of interest is where $\Xi_i$ is either a finite or infinite box in $\R^n$ with faces parallel to the coordinate axes, or an intersection of such a box with a linear half space, in which case it is easy to pass between the primal and dual descriptions. Note 
however that the dual description is preferable, as the description of a box in $\R^n$ requires only $2n$ linear inequalities, while the primal description requires $2^n$ extreme vertices.\\
Consider the problem
\begin{align}\label{eq:ConstraintsViaBoundsOnIntegralsPrimal}
\sup_{\calF\in\cM_{\fF}}\,&\int_{\Phi}h(x)\diff \calF(x)\\
\text{s.t. }&\int_{\Phi}\phi_{s}(x)\diff\calF(x)\leq a_{s},\quad(s=1,\dots,M),\nonumber\\
\text{s.t. }&\int_{\Phi}\psi_{t}(x)\diff\calF(x)= b_{t},\quad(t=1,\dots,N),\nonumber\\
&\int_{\Phi}1\diff\calF(x)=1,\nonumber\\
&\calF\geq 0,\nonumber
\end{align}
where the test functions $\psi_t$ are piecewise linear on the partition $\Phi=\bigcup_{i=1}^k\Xi_i$, and where $-h(x)$ and the test functions $\phi_s$ are piecewise linear on the infinite polyhedra of the partition, and either jointly linear, concave, or convex on the finite polyhedra (i.e., polytopes) of the partition. The dual of \eqref{eq:ConstraintsViaBoundsOnIntegralsPrimal} is
\begin{align}\label{eq:ConstraintsViaBoundsOnIntegralsDual}
\inf_{(y,z)\in\R^{M+N+1}}\,&\sum_{s=1}^{M}a_s y_s+\sum_{t=1}^{N}b_t z_t + z_0,\\
\text{s.t. }&\sum_{s=1}^{M}y_s\phi_s(x)+\sum_{t=1}^{N}z_t\psi_t(x)+z_0\1_{\Phi}(x)
-h(x)\geq 0,\; (x\in\Phi),\nonumber\\
&y\geq 0.\nonumber
\end{align}

Note that \eqref{eq:ConstraintsViaBoundsOnIntegralsPrimal} is a semi-infinite programming problem with infinitely many variables and finitely many constraints, while \eqref{eq:ConstraintsViaBoundsOnIntegralsDual} is a semi-infinite programming problem with finitely many variables and infinitely many constraints. However, the constraint of the problem \eqref{eq:ConstraintsViaBoundsOnIntegralsDual} can be rewritten as co-positivity requirements over the polyhedra $\Xi_i$, 
\begin{equation*}
\sum_{s=1}^{M}y_s\phi_s(x)+\sum_{t=1}^{N}z_t\psi_t(x)+z_0\1_{\Phi}(x)
-h(x)\geq 0,\quad(x\in\Xi_i), \quad(i=1,\dots,k). 
\end{equation*}
Next we will see how these co-positivity constraints can be handled numerically, often by relaxing all but finitely many constraints.\\
In what follows, we will use the notation
\begin{equation*}
\varphi_{y,z}(x)=\sum_{s=1}^{M}y_s\phi_s(x)+\sum_{t=1}^{N}z_t\psi_t(x)+z_0-h(x). 
\end{equation*}

\subsubsection{Piecewise Linear Test Functions}\label{piecewise linear}

The first case we discuss is when $\phi_s|_{\Xi_i}$ and $h|_{\Xi_i}$ are jointly linear. Since we 
furthermore assumed that the functions $\psi_t|_{\Xi_i}$ are linear, there exist vectors $v^i_s\in\R^n$, 
$w^i_t\in\R^n$, $g^i\in\R^n$ and constants $c^i_s\in\R$, $d^i_t\in\R$ and $e^i\in\R$ such that 
\begin{align*}
\phi_s|_{\Xi_i}(x)&=\langle v^i_s, x\rangle+c^i_s,\\
\psi_t|_{\Xi_i}(x)&=\langle w^i_t, x\rangle+d^i_t,\\
h|_{\Xi_i}(x)&=\langle g^i, x\rangle+e^i. 
\end{align*}
The copositivity condition 
\begin{equation*}
\sum_{s=1}^{M}y_s\phi_s(x)+\sum_{t=1}^{N}z_t\psi_t(x)+z_0\1_{\Phi}(x)
-h(x)\geq 0,\quad(x\in\Xi_i) 
\end{equation*}
can then be written as 
\begin{multline*}
\langle f^i_j, x\rangle \geq \ell^i_j,\quad (j=1,\dots,k_i)\Longrightarrow\\
\left\langle\sum_{s=1}^{M}y_s v^i_s+\sum_{t=1}^{N}z_t w^i_t-g^i\;,\; x\right\rangle\geq e^i-\sum_{s=1}^{M}y_s c^i_s - \sum_{t=1}^{N}z_t d^i_t - z_0. 
\end{multline*}
By Farkas' Lemma, this is equivalent to the constraints 
\begin{align}
\sum_{s=1}^{M}y_s v^i_s+\sum_{t=1}^{N}z_t w^i_t-g^i&=\sum_{j=1}^{k_i}\lambda^i_j f^i_j,\label{c1}\\
e^i-\sum_{s=1}^{M}y_s c^i_s - \sum_{t=1}^{N}z_t d^i_t - z_0&\leq\sum_{j=1}^{k_i}\lambda^i_j \ell^i_j,\label{c2}\\
\lambda^i_j&\geq 0,\quad(j=1,\dots,k_i),\label{c3}
\end{align}
where $\lambda^i_j$ are additional auxiliary decision variables.\\
Thus, if all test functions are linear on all polyhedral pieces $\Xi_i$, then the dual (D) can be solved as a linear programming problem with $M+N+1+\sum_{i=1}^k k_i$ variables and $k(n+1)$ linear constraints, 
plus bound constraints on $y$ and the $\lambda^i_j$. More generally, if some but not all polyhedra correspond to jointly linear test function pieces, then jointly linear pieces can be treated as discussed above, while other pieces can be treated as discussed below. 
Details of effective solution methods of the above constructed linear optimisation problem can be found in \cite{hauser}.\\
In the next section we show that the problem of model-free option pricing with piece-wise linear payoffs can be approximated in the framework of this model. The numerical experiments we present in Section \ref{numericalExperiments} show that in some cases the model solves the problem exactly.

\section{Numerical Implementation}
In practice we are not given European option prices for every strike, hence we cannot have exact knowledge of marginal distributions. Moreover, for some maturities the number of traded liquid products can be so small that any interpolation of option prices will contain additional errors. To avoid introducing additional estimation errors in the model, and to use only the information given by the market, we replace constraints for the marginal distributions by constraints which represent calculation of vanilla option values.\\
Taking into consideration the above observations, we construct the following problem
\begin{align}
\quad\sup_{\calF}\,&\int_{\R^n}h(x)\diff\calF(x)\nonumber\\
\text{s.t. }&\int_{\R^n}(x_i - k_i^j)^+\diff\calF(s)\leq \bar C_i^j,\quad i = 1\dots, n,\quad j = 1, \dots, m_j\nonumber\\
&\int_{\R^n}(x_i - k_i^j)^+\diff\calF(s)\geq \underline C_i^j,\quad i = 1\dots, n,\quad j = 1, \dots, m_j\nonumber\\
&\int_{\R^n}\phi_i(s)\diff\calF(s) \leq \bar \delta_i,\quad i = 1, \dots, N_{\phi}\nonumber\\
&\int_{\R^n}\phi_i(s)\diff\calF(s) \geq \underline \delta_i,\quad i = 1, \dots, N_{\phi}\nonumber\\
& \int_{\R^n}\1_{z\leq s|_{k}\leq y}(z, y, s)(x_{k+1} - x_k)\diff\calF(s) = 0,\quad(y,z\in\R^k,k<n)\\
&\int_{\R^n}\1\diff\calF(s) = 1,\nonumber\\
&\calF\geq 0,\nonumber
\end{align}
where $(k_i^j)_{j=1}^{m_i}$ are the strikes of available in the market European call options, $(\bar C_i^j)_{j=1}^{m_i}$ and $(\underline C_i^j)_{j=1}^{m_i}$ are their respective market bounds of prices (bid-ask prices) and $m_i$ is the number of such options.\\
Except of call or put options we might observe other options in the market. Let them have payoff functions $\phi_i(s)$ and bid-ask prices $\underline \delta_i$ and $\bar\delta_i$, $i = 1, \dots, N_{\phi}$. For each of these additional observed options we added $2$ more inequality constraints.\\
Note that the constraints representing call option prices fit our model (the functions are piecewise linear). To render the martingale constraints compatible with our model, we approximate them by requiring that \eqref{eq:martingaleEquivalent} hold for only finitely many boxes in $\R^n$ of our choosing.\\
We impose a lower and upper bound on the stock price $X_i$ for each $t_i$, $i = 1, \dots, n$. $0$ will represent a natural lower limit, and for the upper limit, we can take a large enough value that the stock price is statistically unlikely to exceed. We also make sure the strikes are included in the set of points in our discretisation.\\
Let $0 = d_i^1 < d_i^2 < \dots < d_i^{n_i}$ be the set of discretisation points for $X_i$. Denote $\Phi:=\left\{s:d_i^1 \leq x_i\leq d_i^{n_i}, (i = 1, \dots, n)\right\}$.\\
Then, the martingale constraint for $k = 1$ is
\begin{equation*}
\int_{\R^n}\1_{\{x\in\R^n:d_i^{j_i}\leq x_i\leq d_i^{\ell_i},(i = 1, \dots, k)\}}(s)(x_{k + 1} - x_k)\diff\calF(s) = 0, \quad(1\leq j_i\leq\ell_i\leq n_i, \forall i).
\end{equation*}
For each $k$ we thus have $M_k := \prod_{i=1}^k (n_i - 1)$ equality constraints, which are all represented as integral inequalities with piecewise linear functions.\\
We saw that after discretisation all our constraints fit the model described in Section \ref{modelGeneral}. Further, if the payoff functions $h(s)$ and $\phi_i(s), \quad i = 1, \dots, N_{\phi}$ are piecewise linear (put/call or binary barrier options, Asian options, lookback options etc.), then we can solve the dual problem as described in Section \ref{modelGeneral}. If the payoff function $h(s)$ is concave (as for options on realised variance), then we can find the upper bound with our method, and if it is convex, we find the lower bound. For all other cases we have to make approximations, for example, first order Taylor approximation.\\
Now we can write the dual problem as follows:
\begin{eqnarray*}
\inf_{\underline y, \bar y, \underline z, \bar z, t, w}&&\sum_{i = 1}^{n}\sum_{j=1}^{n_i}(\bar y_i^j \bar C_i^j-\underline y_i^j \underline C_i^j) + \sum_{i=1}^{M_{\phi}}(\bar z_i\bar\delta_i - \underline z_i\underline\delta_i) + w\\
\text{s.t. }&&\sum_{i = 1}^{n}\sum_{j=1}^{n_i}(\bar y_i^j-\underline y_i^j)(x_i - K_i^j)^+\sum_{i=1}^{M_{\phi}}(\bar z_i - \underline z_i)\phi_i(s)+\\
&&\sum_{k=1}^{n-1}\sum_{\tau\in I_k}\1_{p_\tau\leq s|_k\leq p_{\tau+1}}(s)(x_{k+1} - x_k) + w\geq h(s),\quad \forall s\in\Phi.
\end{eqnarray*}

Interested readers can find efficient algorithms to solve the above described optimisation problems in the DPhil thesis of the second author \cite{sergey}. Since the size of the linear optimisation problems grows rapidly with the dimension, memory cost becomes an important issue. For this situation it is useful to implement the delayed column generation method of the simplex algorithm. Instead of creating the whole matrix of the problem, the delayed column generation algorithm gradually adds columns to the tableau while the simplex algorithm takes iterates. In practice, one does not need to consider all the columns in order to achieve the optimal solution to a desired accuracy. Moreover, since a large number of iterations works on just a small number of columns, the method gains a great advantage in saving time by not using the secondary memory.\\
Arbitrage-free pricing is often used as a means to identify arbitrage opportunities. Suppose some path-dependent option with payoff $h(s)$ is traded in the market with price $h^M$. By implementing our model for this option we find no-arbitrage bounds $[L, U]$. If $h^M\notin [L, U]$, then there is an arbitrage in the market (the dual problem also provides with the strategy to make risk-free money in this case, but we will not go into details in this work, interested reader can refer to \cite{beig}). Hence, often we are interested in finding out whether, for a given $h^M$, $L < h^M < U$. For this particular purpose, the detailed description of the implementation of Nesterov's algorithm to linear optimisation problem can be found in \cite{sergey}.\\
\begin{remark}
Note that one can introduce additional constraints into the model depending on additional assumptions one wishes to make on the behaviour of the stock. For example, it is easy to include lower and/or upper bound on the volatility of the stock.
\end{remark}

\section{Payoffs as Piece-wise Linear Functions}
In this section we describe different payoffs that are piece-wise linear functions. Obviously, call and put options fit to this category.
\subsection{Barrier Options}
For given lower and upper barriers $B_1<B_2$ (if the option is defined for one barrier only, e.g. down and out, up and in, then we suppose the other barrier is $-\infty$ or $\infty$ depending on the context), barrier options have payoff
\begin{equation*}
h(s) = H(X)\prod_{i = 1}^n\1_{B_1\leq X_i \leq B_2},
\end{equation*}
where $H(X)$ is the final payoff. For digital barrier options we have $H(S) = 1, \quad \forall S\in\R^n$ and for call and put barrier options with strike $K$ we have $H(S) = (X_n - K)^+$ and $H(S) = (K - X_n)^+$ respectively.\\
It is obvious that these options have piece-wise linear payoffs.

\subsection{Lookback and Asian Options}
Lookback options have payoffs that depend on the $X_{max} = \max\{X_1, \dots, X_n\}$ and/or $X_{min} = \min\{X_1, \dots, X_n\}$. Asian options are defined as having payoff functions at maturity which depend on the average value $X_A = \frac{1}{n}\sum_{i=1}^nX_i$ of the stock price over the path. 
 Table \ref{fig:OptionPayoofs} describes payoffs for different lookback and Asian options.
\begin{table}[htbp]
\caption{Lookback and Asian Option Payoffs}
\begin{center}
\begin{tabular}{|c|c|c|}
\hline
 & \multicolumn{1}{l|}{Call} & \multicolumn{1}{l|}{Put} \\ \hline
Lookback Fixed & $\max(X_{max} - K, 0)$ & $\max(K - X_{min}, 0)$ \\ \hline
Lookback Floating & $X_n - X_{min}$ & $X_{max} - X_n$\\ \hline
Asian Fixed & $\max(X_A - K, 0)$ & $\max(K - X_A, 0)$ \\ \hline
Asian Floating & $X_n - X_A$ & $X_A - X_n$ \\ \hline
\end{tabular}
\end{center}
\label{fig:OptionPayoofs}
\end{table}
The following lemma proves that all payoffs described in the Table \ref{fig:OptionPayoofs} are piece-wise linear functions on $\Phi$.
\begin{lemma}
Let $\Phi$ be a polytope in $\R^n$ and be further divided into finitely many polytopes $\Phi=(\bigcup_{i=1}^M\Pi_i)$. Let function $f:\Phi\rightarrow \R$ and $g : \Phi\rightarrow \R$ be piece-wise linear on these polytopes. Then, the functions $\max\{f(x), g(x)\}$ and $\min\{f(x), g(x)\}$ are also piece-wise linear on some finite partitioning of $\Phi = (\bigcup_{i=1}^{\bar M}\bar\Pi_i)$.
\end{lemma}

\section{Numerical Experiments}\label{numericalExperiments}
In this section we present numerical results for barrier options. We consider $2$ and $3$ dimensional cases. Numerical solutions of the linear programming problem are carried out with interior point methods using SeDuMi and SDPT3 solvers on MATLAB. For higher-dimensional or finer discretisation problems it is better to use Nesterov's smoothening technique for non-smooth functions, \cite{nesterov2, sergey}.\\
First we obtain our numerical results for simulated market data, i.e. we compute call option prices with Black-Scholes formulae for different time horizons. Numerical experiments on simulated market data is of particular interest because in this case we can also compute the true prices of our objective derivative (i.e. the one that we aim to find upper and lower bounds on). Hence, we can compare the price of the derivative in Black-Scholes model with the lower and upper bounds obtained by our model.\\
\begin{table}[htbp]
\begin{center}
\begin{tabular}{|c|c|c|c|}
\hline
 & \multicolumn{1}{l|}{Lower} & \multicolumn{1}{l|}{Upper}& \multicolumn{1}{l|} {Black-Scholes}\\ \hline
Double no-touch digital & 0.282622 & 0.612447 & 0.4232\\ \hline
Double no-touch call & 0 & 0.527483 &  0.202\\ \hline
Variance Swap & 0 & 0.187674 & 0.0919\\ \hline
\end{tabular}
\end{center}
\caption{Lower and Upper Bounds on 2-time step barrier options compared to Black Scholes value.}
\label{fig:2dBarrierOptions}
\end{table}
\begin{table}[htbp]
\begin{center}
\begin{tabular}{|c|c|c|c|}
\hline
& \multicolumn{1}{l|}{Lower} & \multicolumn{1}{l|}{Upper}& \multicolumn{1}{l|}{Black-Scholes} \\ \hline
Double no-touch digital & 0.0610184 & 0.533453 & 0.2732\\ \hline
Double no-touch call & 0 & 0.43506 & 0.1\\ \hline
Variance Swap & 0 & 0.326711 &0.1369 \\ \hline
\end{tabular}
\end{center}
\caption{Lower and Upper Bounds on 3-time step barrier options compared to Black Scholes value.}
\label{fig:3dBarrierOptions}
\end{table}
\begin{table}[htbp]
\begin{center}
\begin{tabular}{|c|c|c|c|}
\hline
& \multicolumn{1}{l|}{Barrier Call} & \multicolumn{1}{l|}{Barrier Digital}& \multicolumn{1}{l|}{Variance Swap} \\ \hline
2 & 0.554805	& 0.90002 &	1.21195\\ \hline
3 & 0.485545	& 0.855626 &	1.24816\\ \hline
\end{tabular}
\end{center}
\caption{Lower bounds on variance swap, double no-touch barrier digital and cal options for 2 and 3 time-steps. The results are obtained using market data for Apple stock. The stock price is 115 and the barriers are 95 and 125.}
\label{fig:marketData}
\end{table}
Tables \ref{fig:2dBarrierOptions} and \ref{fig:3dBarrierOptions} show numerical results for different barrier options for $2$ and $3$ time steps respectively. To obtain these results we only used market given strikes as discretisation points, i.e. for all $i = 1, \dots, n$ and $j = 1, \dots, n_i$, $d_i^j = k_i^j$.\\
We carried out the experiments for simulated option prices with spot price $50$, volatility $30\%$ and time horizons with time steps $6$ months. Strikes are chosen at all even numbers between $30$ and $60$ inclusive, while lower and upper bounds were chosen as $34$ and $56$ respectively.\\
In Table \ref{fig:marketData} we present lower bounds on double no-touch call, digital and variance swap options using market data for call options on Apple stock (APPL). The spot price is $115$ and the barriers considered are $95$ and $125$. We used call option prices for maturities $20/12/14$, $20/02/15$ and	$17/04/15$. The calculation date is $20$ November $2014$.\\
Note that the third option considered here is the variance swap option, which has neither piecewise linear nor piecewise concave payoff. In this case, we approximate the payoff function by a piecewise linear function with a linear restriction over each polytope.\\
After carrying out some numerical experiments, we observe that as long as prices of call options for the strikes that coincide with barriers of our options are available, any finer discretisation than the one described above leads to no further tightening of the bounds, that is, the solution of the optimisation problem remains the same. This observation, a proof of which is the subject of further research, implies that our method yields the tightest possible bounds that are compatible with the given market data. This is to say that for any other bound $l > L$ and $u<U$ one can construct a martingale measure for stock prices such that under this measure we get prices for our barrier options that are out of these new bounds $l$ and $u$, where $L$ and $U$ are the lower and upper bounds obtained by our model.\\
Taking into consideration the observation above, in the case when we don't have barriers as strikes in our market data, we should make some interpolation at the beginning to find approximate call option prices for the barriers as strikes and then construct our model. This approach is better than solving the problem by making a finer approximation.\\
\begin{remark}
Note that in practice information on call option prices provided by the market and the martingale constraint produce very wide bounds, so practically one will never observe an arbitrage, i.e. the price of the objective option will always be within the bounds. But as can be observed in the results that use call option prices obtained from Black-Scholes model, the real value of the option is approximately the median of the lower and upper bounds. 
\end{remark}
\section{Conclusions}
In this work we presented how the general duality results from \cite{hauser} can be applied to model-free path-dependent option pricing. We also introduced the so-called 'constraints via integrals of piece-wise concave functions' to solve the resulting semi-infinite dimensional optimisation problems numerically. Moreover, our numerical experiments provide strong evidence that the bounds obtained via our method are the tightest possible bounds compatible with the observed market prices and the no-arbitrage assumption.
\newpage

\bibliographystyle{plain}  %use the plain bibliography style

\end{document}